\documentclass[cmp,draft,numbook]{svjour}
\usepackage{amsmath,amssymb,verbatim,tikz-cd}

\spnewtheorem*{cor}{Corollary}{\bf}{\it}
\spnewtheorem*{lem}{Lemma}{\bf}{\it}
\spnewtheorem*{pro}{Proposition}{\bf}{\it}
\spnewtheorem*{teo}{Theorem}{\bf}{\it}

\newcommand\al{\alpha}
\newcommand\be{\beta}
\newcommand\dd{\partial}
\newcommand\De{\Delta}
\newcommand\de{\delta}
\newcommand\Ec{\mathcal{E}}

\newcommand\Fc{\mathcal{F}}
\newcommand\FF{\mathbb{F}}
\newcommand\ga{\gamma}
\renewcommand\ge{\geqslant}
\newcommand\La{\Lambda}
\newcommand\la{\lambda}
\newcommand\lcd{,\ldots,}
\renewcommand\le{\leqslant}

\newcommand\QQ{\mathbb{Q}}
\newcommand\Sg{\mathfrak{S}}
\newcommand\si{\sigma}
\newcommand\ts{\hspace{0.75pt}}
\newcommand\Uc{\mathcal{U}}
\newcommand\Zc{\mathcal{Z}}

\begin{document}

\title{Cherednik operators and
\\ 
Ruijsenaars\ts-Schneider model at infinity\\}
\titlerunning{Cherednik operators at infinity}

\author{M.\,L.\,Nazarov and E.\,K.\,Sklyanin}
\authorrunning{Nazarov and Sklyanin}

\institute{Department of Mathematics, University of York, 
York YO10 5DD, United Kingdom}

\date{}

\maketitle


\thispagestyle{empty} 

\begin{abstract}
Heckman introduced $N$ operators on the space of polynomials in
$N$ variables, such that these operators form a covariant set
relative to permutations of the operators and variables, and such that Jack
symmetric polynomials are eigenfunctions of the power sums of these
operators. We introduce the analogues of these $N$ operators 
for Macdonald symmetric polynomials, by using Cherednik operators. 
The latter operators pairwise commute, and Macdonald polynomials are
eigenfunctions of their power sums. We compute the limits of our
operators at $N\to\infty\,$. These limits yield a 
Lax operator for Macdonald symmetric functions. 
\end{abstract}



\section*{Introduction}

The present article is a continuation of our works \cite{NS3,NS4}.
By using the Lax operator formalism, in \cite{NS3}  
we constructed a family of pairwise commuting operators such that 
the Jack symmetric functions of the infinitely many variables 
$x_1\ts,x_2\ts,\,\ldots$ are their eigenfunctions.
We expressed these commuting operators in terms of
the power sum symmetric functions
$x_1^n+x_2^n+\ldots$ where $n=1,2,\,\ldots\ $.

The Jack symmetric functions can be regarded as limiting cases
of Macdonald symmetric functions. The latter depend on the same 
variables $x_1\ts,x_2\ts,\,\ldots$ and also on
two parameters $q,t\,$. In \cite{NS4} we extended the results of
\cite{NS3} to the latter setting. 
In particular, by again using the Lax operator formalism
we constructed a family of pairwise commuting operators such~that 
the Macdonald symmetric functions are their eigenfunctions.
In \cite{NS4} we expressed these commuting operators in terms of
the Hall\ts-Littlewood symmetric functions
of the variables $x_1,x_2,\,\ldots\ $
and of the parameter $t\ts$.
These expressions involve only
the Hall\ts-Littlewood symmetric functions  
corresponding to the partitions with one part,
see Subsection \ref{sec:111} here.

\enlargethispage{30mm}

Shortly after \cite{NS3} was published,
A.\,N.\,Sergeev and A.\,P.\,Veselov communicated to us
their remarkable works \cite{SV1,SV2} where in particular they found 
essentially the same commuting operators as we did in \cite{NS3}.
Their approach was different however. 
They first computed the limits at $N\to\infty$ 
of the Heckman operators~\cite{H} acting on all polynomials 
in the variables $x_1\lcd x_N\,$. These $N$
operators do not commute in general.
But the restrictions of the power sums of these $N$ operators
to the space of symmetric polynomials do commute. Moreover,
Jack symmetric polynomials are eigenfunctions of these restrictions.
Jack symmetric functions are then eigenfunctions
of the limits of these restrictions at $N\to\infty\,$.

\newpage
\voffset-5mm

Here we extend this approach from Jack to 
Macdonald symmetric functions. It has been discovered 
by I.\,V.\,Cherednik \cite{C}
that the Macdonald polynomials in the variables
$x_1\lcd x_N$ are eigenfunctions of power sums of some
$N$ commuting operators, acting on all polynomials in these variables.
These operators are called the Cherednik operators, see our
Subsection \ref{sec:13} for their definition.
It has been also known \cite{SV0} that the Cherednik operators
have limits at $N\to\infty\,$. 
However, explicit expressions for these limits are unknown.
We offer a solution to this open problem
by firstly introducing for the Macdonald polynomials
the analogues of non-commuting Heckman operators,
see our Subsection \ref{sec:15}. 
These analogues act on the rational functions
of $x_1\lcd x_N$ and  are denoted by $Z_1\lcd Z_N\,$.
They are related to the Cherednik operators by Proposition \ref{sec:15}.
The principal property~of the operators $Z_1\lcd Z_N$ is stated 
as Theorem~\ref{sec:16}, see also Proposition \ref{sec:13}.

An explanation is
needed here regarding our scheme of referring to lemmas,
propositions, theorems and corollaries. When referring to these, 
we indicate the subsections where they respectively appear. 
There is no more than one of each of these in every subsection, 
so our scheme should cause no confusion. 
For example, Proposition \ref{sec:13} 
is the only proposition that appears 
in Subsection~\ref{sec:13}.

In Subsection \ref{sec:17} we reformulate Theorem \ref{sec:16}
by introducing a certain $N\times N$ matrix $\Zc$ with operator entries
acting on the rational functions of the variables $x_1\lcd x_N\,$. 
It is closely related to the classical Lax matrix of the 
trigonometric Ruijsenaars\ts-Schneider model \cite{RS}. 
To be precise, let $\ga_{\ts i}$ be the operator defined by \eqref{gai}.
In the classical limit $q\to1\,$, when the canonical commutation relation 
$\ga_{\ts i}\,x_i=q^{\ts-1}\ts x_i\,\ga_{\ts i}$ degenerates 
to the Poisson bracket $\{\ts\ga_{\ts i}\ts,x_i\ts\}=-\,\ga_{\ts i}\,x_i\,$, 
our $\Zc$ degenerates to this Lax matrix 
up to a change of variables and up to conjugation by a diagonal matrix.
In the same classical limit, 
the Macdonald determinant \eqref{sedeb} degenerates to the 
characteristic determinant of the Lax matrix.
Thus 
we have found a way to derive a quantum analogue of this  
Lax matrix directly from the Cherednik operators.

It is well known how the quantum 
Hamiltonians \cite{R} of the trigonometric
Ruijsenaars\ts-Schneider model are related
to the Macdonald operators \eqref{mk}, see for instance \cite{K}.
But our generating series \eqref{Ham} for quantum Hamiltonians 
differs from the \text{Macdonald} determinant and is new.
A similar resolvent-type expression was used 
for the quantum Calogero-Sutherland model in \cite{UWH}.
The eigenstates of the latter model are the Jack symmetric polynomials.
The limit at $N\to\infty$~of the generalisation of that model to 
particles with spin has been studied in \cite{KMS}
as another extension of our work \cite{NS3}.

Following the approach of \cite{SV1,SV2}
in Subsection \ref{sec:31} of our article
we compute the limits at $N\to\infty$ of the operators $Z_1\lcd Z_N\,$.
Then we also compute the limit of the restriction 
of the operator sum \eqref{uzi} appearing in Theorem~\ref{sec:16}
to the space of symmetric polynomials in 
$x_1\lcd x_N\,$. This limit is a formal power series 
in another variable $u$ with operator coefficients acting on the
symmetric functions of $x_1,x_2,\,\ldots\ $.
After renormalisation and a change of the variable $u\,$, 
this limit becomes
the same generating series of the 
pairwise commuting operators as we constructed in \cite{NS4}.
For details, see Subsections \ref{sec:32} and \ref{sec:33} here.

In this article we generally keep to
the notation of the book \cite{M} 
for symmetric functions. When using results from \cite{M}
we simply indicate their numbers within the book.
For example, the statement (6.9) from Chapter I of the book
will be referred to as [I.6.9] assuming it is from~\cite{M}. 


\newpage

\section{Symmetric functions}
\label{sec:1}


\subsection{Standard symmetric functions}
\label{sec:11}

Fix a field $\FF\,$. For any positive integer $N\ge1$ 
denote by $\La_N$ the $\FF$-algebra of symmetric 
polynomials in $N$ variables $x_1\lcd x_N\,$.
The algebra $\La_N$ is graded by the polynomial degree.
The substitution $x_N=0$ defines
a homomorphism $\La_N\to\La_{N-1}$ preserving the degree.
Here $\La_{\ts0}=\FF\,$.
The inverse limit of the sequence
$$
\La_1\leftarrow\La_2\leftarrow\ldots
\hspace{-8pt}
$$ 
in the category of graded algebras is denoted by $\La\,$.
Note that we get a canonical homomorhism $\La\to\La_N\,$. 
The elements of the algebra
$\La$ are called \textit{symmetric functions\/}.
Following \cite{M} we now will introduce some standard bases of $\La\,$.
 
Let $\la=(\,\la_1,\la_2,\ldots\,)$ be any partition of $\,0,1,2,\ldots\,\,$. 
The number of non-zero parts is called the {\it length\/} of 
$\la$ and is denoted by $\ell(\la)\,$. 
If $\ell(\la)\le N$ then
the sum of all distinct 
monomials obtained by permuting the $N$ variables in
$x_1^{\,\la_1}\ldots x_N^{\,\la_N}$  is denoted by
$m_\la(x_1\lcd x_N)\,$. 
The symmetric polynomials $m_\la(x_1\lcd x_N)$ with
$\ell(\la)\le N$ form a basis of the vector space $\La_N\,$. 
By definition, for $\ell(\la)\le N$
\begin{equation*}
m_\la(x_1\lcd x_N)\ =
\sum_{1\le i_1<\ldots<i_k\le N}
\ \sum_{\si\in\Sg_k}\ c_\la^{-1}\ 
x_{i_{\si(1)}}^{\,\la_1}\ldots x_{i_{\si(k)}}^{\,\la_k}
\end{equation*}
where we write $k$ instead of $\ell(\la)\,$. Here
$\Sg_k$ is the symmetric group permuting the numbers $1\lcd k$
and
\begin{equation*}
c_\la=k_1!\ts\,k_2!\,\ldots
\end{equation*}
if $k_1,k_2,\ldots$ are the respective multiplicites of the parts $1,2,\ldots$
of $\la\,$. Further,
\begin{equation*}
m_\la(x_1\lcd x_{N-1},0)\,=\,
\left\{
\begin{array}{cl}
m_\la(x_1\lcd x_{N-1})
&\quad\textrm{if}\quad\,\ell(\la)<N\,;
\\[2pt]
0
&\quad\textrm{if}\quad\,\ell(\la)=N\,.
\end{array}
\right.
\end{equation*}
Hence for any fixed partition $\la$ the sequence of polynomials
$m_\la(x_1\lcd x_N)$ with $N\ge\ell(\la)$ 
has a limit in $\La\,$. This limit is called
the \textit{monomial symmetric function\/}
corresponding to $\la\,$. Simply omitting the variables,
we will denote the limit by $m_\la\,$.
With $\la$ ranging over all partitions of $0\ts,1\ts,2\ts,\ts\ldots$ 
the symmetric functions $m_\la$ form a basis of the vector space $\La\,$.
Note that if $\ell(\la)=0$ then we set $m_\la=1\,$.

We will be also using another standard basis of the vector space $\La\,$.
For each $n=1,2,\ldots$ denote $p_n(x_1\lcd x_N)=x_1^n+\ldots+x_N^n\,$.
When the index $n$ is fixed
the sequence of symmetric polynomials $p_n(x_1\lcd x_N)$ with
$N=1,2,\ldots$ has a limit in $\La\,$, called
the \textit{power sum symmetric function} of degree $n\,$.
We will denote it by $p_n\,$. 
More generally, for any partition $\la$ put
\begin{equation*}
\label{pla}
p_\la=p_{\la_1}\ldots p_{\la_k}
\end{equation*}
where $k=\ell(\la)$ as above. The elements $p_\la$
form another basis of $\La\,$. Equivalently,
the elements $p_1,p_2,\ldots$ are
free generators of the commutative algebra $\La$ over $\FF\,$.

In this article we will be using 
the \textit{natural ordering\/} of partitions.
By definition, here $\la\ge\mu$ if
$\la$ and $\mu$ are partitions of the same number and
$$
\la_1\ge\mu_1,\ \,
\la_1+\la_2\ge\mu_1+\mu_2,\ \,
\ldots\ \,.
$$
This is a partial ordering.
Note that by [I.6.9] any monomial symmetric function $m_\mu$
is a linear combination of the symmetric functions 
$p_\la$ where $\la\ge\mu\,$.

Define a bilinear form $\langle\ ,\,\rangle$ on 
$\La$ by setting for any two partitions $\la$ and $\mu$
\begin{equation}
\label{schurprod}
\langle\,p_\la\,,p_\mu\ts\rangle=k_\la\ts\de_{\la\mu}
\quad\text{where}\quad
k_\la=1^{\ts k_1}k_1!\ts\,2^{\,k_2}k_2!\ts\,\ldots
\end{equation}
in the above notation. This form is obviously symmetric
and non-degenerate. We will indicate by the superscript ${}^\perp$
the operator conjugation relative to this form. In particular, by 
\eqref{schurprod}
for the operator conjugate to the multiplication in~$\La$ by $p_n$
with $n\ge1$ we have
\begin{equation}
\label{pperp}
p_n^{\ts\perp}=n\,\dd/\dd\,p_n\,.
\end{equation}

Next put
$$
e_n(x_1\lcd x_N)\ =
\sum_{1\le i_1<\ldots<i_n\le N}
x_{i_1}\ldots x_{i_n}\,.
$$
For any fixed $n$
the sequence of the symmetric polynomials $e_n(x_1\lcd x_N)$ with
$N=1,2,\ldots$ has a limit in $\La\,$, denoted by $e_n$
and called
the \textit{elementary symmetric function\/} of degree $n\,.$
We will also use a formal power series in another variable~$v\,$,
\begin{equation}
\label{ev}
E(v)=1+e_1\ts v+e_2\ts v^{\ts2}+\ldots\,=\,\prod_{i\ge1}\,(1+x_i\ts v)\,.
\end{equation}
By taking logarithms of the left and right hand side
of the above display and~then exponentiating,
\begin{equation}
\label{euexp}
E(v)\,=\,\exp\,\Bigl(\ts-
\sum_{n\ge1}\,
\frac{p_n}n\ts (\ts-\ts v\ts)^n
\ts\Bigr)\,.
\end{equation}

Also put
$$
h_n(x_1\lcd x_N)\,=\,
\sum_{\ell(\la)\le N}\,
m_\la(x_1\lcd x_N)
$$
where the sum is taken over partitions $\la$ of $n\,.$
Then the sequence of symmetric polynomials $h_n(x_1\lcd x_N)$ 
with $N=1,2,\ldots$ has a limit in $\La\,$, 
denoted by $h_n$ and called
the \textit{complete symmetric function\/}
of degree $n\,$. By [I.2.6] for the series
\begin{equation*}
H(v)=1+h_1\ts v+h_2\ts v^{\ts2}+\ldots\,.
\end{equation*}
we have the relation
\begin{equation}
\label{eh}
E\ts(\ts-\ts v\ts)\,H(v)=1\,.
\end{equation}
Hence \eqref{euexp} implies
\begin{equation}
\label{huexp}
H(v)\,=\,\exp\,\Bigl(\,\,
\sum_{n\ge1}\,
\frac{p_n}n\ts v^{\ts n}
\ts\Bigr)\,.
\end{equation}
The elements $h_1,h_2,\ldots$ as well as the elements $e_1,e_2,\ldots$
are free generators of the commutative algebra~$\La$ over the field $\FF\,$.
We will also use the \textit{vertex operator}
\begin{equation}
\label{hperp}
H^{\ts\perp}(v)=1+h_1^{\ts\perp}\ts v+h_2^{\ts\perp}\ts v^2+\ldots
\,=\,
\exp\,\Bigl(\,\,
\sum_{n\ge1}\,
\frac{p_n^{\ts\perp}}n\ts v^n
\ts\Bigr)\,.
\end{equation}
It follows from \eqref{pperp} and \eqref{hperp} that 
for any $n=1,2,\ldots$ we have the equality
\begin{equation}
\label{vop}
H^{\ts\perp}(v)\ts\,p_n=v^n+\ts p_n\,.
\end{equation}
It also follows from \eqref{pperp} and \eqref{hperp} that 
$H^{\ts\perp}(v):\La\to\La\ts[\ts v\ts ]$ is a homomorphism of 
$\FF$-algebras. See [\ts I.5, Example 29\ts] for both of the last two
statements. Hence by applying $H^{\ts\perp}(v)$ to any symmetric function in 
the variables $x_1,x_2\ts,\,\ldots$ we get the same symmetric 
function but in the variables $v\ts,x_1,x_2\ts,\,\ldots\ \ts$.


\subsection{Hall\ts-Littlewood symmetric functions}
\label{sec:111}

Let $\FF$ be the field $\QQ\ts(t)$ with $t$ a \text{formal} parameter.
The Hall\ts-Littlewood symmetric functions [III.2.11]
are labelled by all partitions of $0\ts,1\ts,2\ts,\ts\ldots$ 
and constitute another remarkable basis of the vector space $\La$ 
over $\FF\,$. 
In the present article we will use
only the elements of this basis corresponding to 
the single part partitions $(1)\ts,(2)\ts,\,\ldots\ $. These elements 
will be denoted by $Q_1\ts,Q_{\ts2}\ts,\ts\ldots$ respectively.
Their generating series is
$$
Q(v)=E(-\ts t\ts v)\ts H(v)=1+Q_1\ts v+Q_2\ts v^{\ts2}+\ldots\,.
$$
By using \eqref{ev} and \eqref{eh} we get the relation 
\begin{equation}
\label{qvprod}
Q(v)\,=\,\prod_{i\ge1}\ts
\frac{\,1-t\,x_i\,v\,}{1-x_i\ts v}
\end{equation}
while by using \eqref{euexp} and \eqref{huexp} we get the relation
\begin{equation}
\label{qvexp}
Q(v)\,=\,\exp\,\Bigl(\,\,
\sum_{n\ge1}
\frac{1-t^{\ts n}\!}{n}\,\ts p_n\ts v^n
\ts\Bigr)\,.
\end{equation}


\subsection{Macdonald symmetric functions}
\label{sec:1111}

Now let $\FF$ be the field $\QQ\ts(q,t)$
where $q$~and~$t$ are formal parameters.
Then define a bilinear form $\langle\ ,\,\rangle_{q,t}$ on $\La$ by setting
\begin{equation}
\label{macprod}
\langle\,p_\la\,,p_\mu\ts\rangle_{q,t}=k_\la\,\de_{\la\mu}\,
\prod_{i=1}^{\ell(\la)}\,\frac{1-q^{\ts\la_i}}{1-t^{\ts\la_i}\,}
\end{equation}
for any partitions $\la$ and $\mu\,$.
This form is symmetric and non-degenerate. If $q=t\,$,
it specializes to the form defined by \eqref{schurprod}.
We will indicate by the superscript ${}^\ast$
the operator conjugation relative to $\langle\ ,\,\rangle_{q,t}\,$. 
In particular, 
by \eqref{pperp} and \eqref{macprod} 
\begin{equation*}
\label{past}
p_n^{\ts\ast}=\frac{1-q^n}{1-t^n}\,p^{\ts\perp}_n
\end{equation*}
for any $n\ge1\,$. Hence by using \eqref{qvexp} we get
\begin{equation*}
\label{qvast}
Q^{\ts\ast}(v)=1+Q_1^{\ts\ast}\ts v+Q_2^{\ts\ast}\ts v^{\ts2}+\ldots
\,=\,\exp\,\Bigl(\,\,
\sum_{n\ge1}
\frac{1-q^n\!}{n}\,\ts p_n^{\ts\perp}\ts v^n
\ts\Bigr)\,.
\end{equation*}
Note that by using \eqref{huexp}, 
the latter identity can be rewritten as
\begin{equation}
\label{qhh}
Q^{\ts\ast}(v)=
H^{\ts\perp}(\ts v\,q\ts)^{\ts-1}
H^{\ts\perp}(v)\,.
\end{equation}
Similarly to $H^{\ts\perp}(v)$ the map 
$Q^{\ts\ast}(v):\La\to\La\ts[\ts v\ts ]$ is a homomorphism of 
$\FF$-algebras.

By [VI.4.7] 
there exists a unique family of elements $P_\la\in\La$ such that
$$
\langle\,P_{\ts\la}\,,P_{\ts\mu}\ts\rangle_{q,t}=0
\quad\text{for}\quad
\la\neq\mu
$$
and such that any $P_{\ts\la}$ equals $m_\la$ 
plus a linear combination of the elements $m_{\ts\mu}$
with $\mu<\la$ in the natural partial ordering.
The elements $P_{\ts\la}\in\La$ are called 
the \textit{Macdonald symmetric functions\/}.

By [VI.4.10] the canonical homomorphism $\La\to\La_N$ maps
$P_{\ts\la}\mapsto0$ if $\ts\ell(\la)>N$.
If $\ell(\la)\le N$ then the image of $P_{\ts\la}\in\La$
under the homomorphism $\La\to\La_N$ is
the \textit{Macdonald symmetric polynomial\/}
usually denoted by 
$P_{\ts\la}(x_1\lcd x_N)\,$.
All these polynomials 
with $\ell(\la)=0\ts,1\ts\lcd N$
form a basis of the vector space $\La_N$ over $\FF\ts$.

\newpage


\section{Cherednik operators}

\subsection{Macdonald operators}
\label{sec:12}

Let $\FF=\QQ\ts(q,t)$ as in Subsection \ref{sec:1111}.
For $i=1\lcd N$
the \textit{inverse q\ts-shift operator\/} $\ga_{\ts i}$ acts on any
rational function 
$f\in\FF\ts(\ts x_1\lcd x_N)$~by
\begin{equation}
\label{gai}
(\ga_{\ts i}\ts f)(x_1\lcd x_N)=f(x_1\lcd q^{\ts-1}x_i\lcd x_N)\,.
\end{equation}
Denote by $\De\ts(\ts x_1\lcd x_N)$ the {\it Vandermonde polynomial\/} 
of $N$ variables
$$
\det\Big[x_i^{\,N-j}\Big]{\phantom{\big[}\!\!}_{i,j=1}^N=
\prod_{1\le i<j\le N}(\ts x_i-x_j)\,.
$$
Put
\begin{equation}
\label{dnu}
D_N(u)=
\De\ts(\ts x_1\lcd x_N)^{-1}\cdot
\det\Big[\,
x_i^{\,N-j}\bigl(\ts 1+u\,t^{\,j-1}\ts\ga_{\ts i}\ts\bigr)
\Big]{\phantom{\big[}\!\!}_{i,j=1}^N
\hspace{-8pt}
\end{equation}
where $u$ is another variable.
The last determinant is defined as the alternated~sum
\begin{equation}
\label{sedeb}
\sum_{\si\in\Sg_N}
(-1)^{\si}\,
\prod_{i=1}^N\,\,\bigl(\,
x_i^{\,N-\si(i)}\bigl(\ts 1+u\,t^{\,\si(i)-1}\ts\ga_{\ts i}\ts\bigr)\bigr)
\end{equation}
where as usual $(-1)^{\si}$ denotes the sign of permutation $\si\ts$.
In every product over $i=1\lcd N$ appearing in the alternated sum
all the operator factors pairwise commute, 
hence their ordering does not matter. 
Note that $D_N(0)=1\,$.

By \eqref{dnu} the $D_N(u)$ is a polynomial of degree $N$ in the variable $u$
with operator coefficients. It also follows 
from \eqref{dnu} that these coefficients~map the space $\La_N$ to itself.
By [VI.3.3] for any $k=1\lcd N$ the coefficient of $D_N(u)$ at $u^k$ equals
\begin{equation}
\label{mk}
\sum_{|I|=k}\,S_I(x_1\lcd x_N)\,\prod_{i\in I}\ts\ga_{\ts i}
\end{equation}
where the sum is taken over all subsets $I$ of $\{1\lcd N\}$ of size $k\,$, 
whereas
$$
S_I(x_1\lcd x_N)=t^{\ts\,k\ts(k-1)/2}\,
\prod_{\substack{i\in I\\j\notin I}}\,
\frac{\,x_i-t\,x_j}{x_i-x_j}\ .
$$

Now for every $k=1\lcd N$ 
consider the restriction of the operator \eqref{mk}
to the space $\La_N\,$. By [VI.4.16]
all these restrictions to $\La_N$ pairwise commute.
They are called the \textit{Macdonald operators\/}.
The Macdonald polynomials $P_{\ts\la}(x_1\lcd x_N)$ 
with $\ell(\la)\le N$ 
make a common eigenbasis of these operators.
By [VI.4.15] the eigenvalue of $D_N(u)$ corresponding to any such eigenvector 
$P_{\ts\la}(x_1\lcd x_N)$ is
\begin{equation}
\label{eigenvalue}
\prod_{i=1}^N\,(\ts1+u\,q^{\ts-\la_i}\ts t^{\,i-1}\ts)\,\ts.
\end{equation}

Note that our definition \eqref{dnu} of the $D_N(u)$ differs from
[VI.3.2] by changing the parameters $q\mapsto q^{\ts-1}$ and
$t\mapsto t^{\ts-1}\,$. However, by [VI.4.14]
the Macdonald polynomials $P_{\ts\la}(x_1\lcd x_N)$ 
are invariant under this change of their parameters.
After this change, we also replaced the variable $X$ 
used in the definition [VI.3.2] by $u\,t^{\,N-1}\ts$. The reasons 
for these alterations will be explained in Subsection \ref{sec:31}.


\subsection{Cherednik operators}
\label{sec:13}

For $i,j=1\lcd N$ with $i\neq j$ introduce the operator acting 
on the vector space $\FF(\ts x_1\lcd x_N\ts)$
\begin{equation}
\label{rij}
R_{\ts ij}=1+\frac{(\ts1-t\ts)\,x_j}{x_i-x_j}\,(\ts1-\si_{ij})=
\frac{x_i-t\,x_j}{x_i-x_j}+
\frac{(\,t-1\ts)\,x_j}{x_i-x_j}\,\si_{ij}
\end{equation}
where $\si_{ij}\in\Sg_N$ acts
by exchanging the variables $x_i$ and $x_j\,$.
It is immediately obvious from the definition \eqref{rij} that
the operator $R_{ij}$ maps polynomials in the variables $x_1\lcd x_N$
to polynomials. Further, one can check that
\begin{equation*}
t\,R_{\ts ij}^{\ts-1}=t+\frac{(\,t-1\ts)\,x_j}{x_i-x_j}\,(1-\si_{ij})=
\frac{t\,x_i-x_j}{x_i-x_j}+
\frac{(\ts1-t\,)\,x_j}{x_i-x_j}\,\si_{ij}\,.
\end{equation*}
The \textit{Cherednik operators\/} $C_1\lcd C_N$
acting on $\FF[\ts x_1\lcd x_N\ts]$ are then defined~by
\begin{equation}
\label{ci}
C_i=t^{\,i-1}\ts 
R_{\ts i,i+1}\ldots R_{\ts iN}\,
\ga_{\ts i}\,
R_{\ts 1i}^{\ts-1}\ldots R_{\ts i-1,i}^{\ts-1}\,.
\end{equation}
These operators pairwise commute.
In general, they do not map the space $\La_N$ to itself.
However, any symmetric polynomial of 
the operators $C_1\lcd C_N$ with the coefficients from the field $\FF$ does.
Moreover by \cite[Subsection 1.3.5]{C} we have

\begin{pro}
The action of $D_N(u)$ on $\La_N$ coincides with that of the product
\begin{equation}
\label{cprod}
\prod_{i=1}^N\,\ts(1+u\,C_i)\,. 
\end{equation}
\end{pro}

In accord with the remark we made at the end of previous subsection,
the operator $C_i$ differs from the operator 
defined by \cite[Equation 1.3.32]{C}
by changing the parameters $q\mapsto q^{\ts-1}$ and
$t\mapsto t^{\ts-1}\,$.
Our normalisation of $C_i$ is also different.


\subsection{Coherence property}
\label{sec:14}

We will use the following property of operators \eqref{ci}.
For $k=1\lcd N-1$ let 
\begin{equation}
\label{cik}
C_1^{\ts(k)}\lcd\ts C_{N-k}^{\ts(k)}
\end{equation}
be the Cherednik operators acting on 
$\FF[\ts x_{k+1}\lcd x_N\ts]$ instead of $\FF[\ts x_1\lcd x_N\ts]\,$.

\begin{lem}
The action of 
\begin{equation}
\label{kprod}
\prod_{i=k+1}^N\!(1+u\,C_i) 
\hspace{-2pt}
\end{equation}
on the space $\La_N$ coincides with the action of 
\begin{equation}
\label{sprod}
\prod_{i=1}^{N-k}\ts(1+u\,t^{\,k}\ts C_i^{\ts(k)})\,. 
\hspace{-20pt} 
\end{equation}
\end{lem}

\begin{proof}
First let us prove by the downward induction on $k=N,N-1\,\lcd 1\ts,0$
that the action of \eqref{kprod}
on the space $\La_N$ coincides with the action of the product
\begin{equation}
\label{rprod}
\prod_{i=k+1}^N\!(1+u\,t^{\,i-1}\ts 
R_{\ts i,i+1}\ldots R_{\ts iN}\,
\ga_{\ts i}\,) 
\end{equation}
where the factors corresponding to the indices $i=k+1\lcd N$ are arranged 
from left to right. If $k=N$ then neither of the products \eqref{kprod}
and \eqref{rprod} has any factors, so the statement to prove is trivial.
Now assume that our statement is already proved for some $k>0\,$.
Consider the product obtained from \eqref{kprod} by replacing
the index $k$ by $k-1\,$. By the induction assumption, the action on $\La_N$
of the so obtained product coincides with that of
\begin{equation}
\label{istep}
(\ts1+u\,C_k)\!
\prod_{i=k+1}^N\!(\ts1+u\,t^{\ts\,i-1}\ts 
R_{\ts i,i+1}\ldots R_{\ts iN}\,
\ga_{\ts i}\,) 
\end{equation}
The last $k-1$ factors of the Cherednik operator $C_k$ appearing 
in \eqref{istep}
$$
R_{\ts 1k}^{\ts-1}\,\lcd R_{\ts k-1,k}^{\ts-1}
$$
commute with $R_{\ts i,i+1}$ and $\ga_{\ts i}$ for any $i=k+1\lcd N\ts$. 
They also act trivially on $\La_N\,$. After removing these $k-1$ factors 
from $C_k$ in \eqref{istep} we get the product
$$
(\ts1+u\,t^{\,k-1}\ts R_{\ts k,k+1}\ldots R_{\ts kN}\,\ga_k\,)\!
\prod_{i=k+1}^N\!(\ts1+u\,t^{\,i-1}\ts 
R_{\ts i,i+1}\ldots R_{\ts iN}\,
\ga_{\ts i}\,)\,. 
$$
Thus we are making the induction step, and our statement is proved 
for any $k\,$.

By using this statement in the particular case when $k=0\,$, the action of 
\eqref{cprod}
on the space $\La_N$ coincides with the action of the product
$$
\prod_{i=1}^N\,\,(\ts1+u\,t^{\,i-1}\ts 
R_{\ts i,i+1}\ldots R_{\ts iN}\,
\ga_{\ts i}\,)\,.
$$
By applying the latter result to the set of 
operators \eqref{cik}
instead of $C_1\lcd C_N$ we obtain that for $0<k<N$ 
the action of \eqref{sprod}
on $\La_N$ coincides with that of
$$
\prod_{i=1}^{N-k}\,(\ts1+u\,t^{\,i+k-1}\ts 
R_{\ts i+k,i+k+1}\ldots R_{\ts i+k,N}\,
\ga_{\ts i+k}\,)\,.
$$
The last displayed product equals \eqref{rprod}
by renaming $i+k$ to $i\,$. 
But we had also proved that the action of \eqref{kprod}
on $\La_N$ coincides with the action of \eqref{rprod}.
\qed
\end{proof}


\subsection{Covariant operators}
\label{sec:15}

For $i,j=1\lcd N$ with $i\neq j$ denote 
\begin{equation}
\label{AB}
A_{\ts ij}=
\frac{\,x_i-t\,x_j}{x_i-x_j}
\quad\text{and}\quad
B_{\ts ij}=\frac{(\,t-1\ts)\,x_j}{x_i-x_j}
\end{equation}  
so that by \eqref{rij}
$$
R_{\ts ij}=A_{\ts ij}+B_{\ts ij}\,\si_{ij}\,.
$$
Define operators $Z_1\lcd Z_N$ acting on $\FF\ts(\ts x_1\lcd x_N)$ 
by setting $Z_i=W_i\ts\ga_{\ts i}$ where
\begin{equation}
\label{wi}
W_i\,=\, 
\prod_{l\neq i}\,A_{\ts il}\,+\,
\sum_{j\neq i}\,B_{\ts ij}\,
\Bigl(\,\prod_{l\neq i,j}A_{\ts jl}\ts\Bigr)\,\si_{ij}\,.
\end{equation}
In general, these operators do not map polynomials
in the variables $x_1\lcd x_N$ to polynomials.  
But by definition, these operators make a \textit{covariant set\/}
relative to the action of the symmetric group $\Sg_N$ by permutations 
of the \text{variables\ts:}
\begin{equation}
\label{coz}
\si^{\ts-1}\ts Z_i\,\si=Z_{\ts\si(i)}
\quad\text{for}\quad
\si\in\Sg_N\,.
\end{equation}
Note that for $N>1$ the operators $C_1\lcd C_N$ on
$\FF[\ts x_1\lcd x_N]$ do not enjoy the covariance property.
On the other hand, our $Z_1\lcd Z_N$ do not commute.

For $k=1\lcd N-1$ let 
$\La_N^{\ts(k)}\subset\ts\FF[\ts x_1\lcd x_N\ts]$
be the subspace of polynomials symmetric in the variables
$x_{k+1}\ts\lcd x_N\,$. Then
$$
\La_N\subset\La_N^{\ts(1)}\subset\ldots\subset\La_N^{\ts(N-1)}
=\ts\FF[\ts x_1\lcd x_N\ts]\,.
$$
Now consider the Cherednik operator $C_1$ acting 
on $\FF[\ts x_1\lcd x_N\ts]\,$.
Our definition of the operator $Z_1$ originates from the 
following proposition. 

\begin{pro}
The actions of the operators $C_1$ and\/ $Z_1$ on $\La_N^{\ts(1)}$ coincide.
\end{pro}

\begin{proof}
We will prove that the action of $C_1$ on $\La_N^{\ts(k)}$ 
coincides with the action of
\begin{equation}
\label{rk}
R_{\ts 12}\ldots R_{\ts 1k}\,\,
\Bigl(\ 
\prod_{k<\ts l\ts\le N}A_{\ts1l}\,+\!
\sum_{k<j\le N}B_{\ts1j}\,
\Bigl(\ \prod_{\substack{k<\ts l\ts\le N\\l\neq j}}A_{\ts jl}\ts
\Bigr)\,\si_{\ts1j}\ts
\Bigr)\,\ga_{\ts1}\,.
\end{equation}
We will use the downward induction on $k=N-1\,\lcd 1\,$.
Our proposition will be then obtained when $k=1\,$.
If $k=N-1$ then by the definition \eqref{ci} we have
$$
C_1=R_{\ts 12}\ldots R_{\ts 1N}\,
\ga_{\ts1}=
R_{\ts 12}\ldots R_{\ts 1,N-1}\,(\ts A_{\ts1N}+B_{\ts1N}\,\si_{\ts1N}\ts)\,
\ga_{\ts1}
$$
as required. Now assume that our statement is proved for some $k>1\,$. Since
\begin{equation}
\label{km}
\La_N^{\ts(k-1)}\subset\La_N^{\ts(k)}
\end{equation}
we then know in particular that the action of $C_1$ on the space
$\La_N^{\ts(k-1)}$ 
coincides with the action of the product \eqref{rk}. 
The latter product can be rewritten as
\begin{gather*}
R_{\ts 12}\ldots R_{\ts 1,k-1}\ \times
\\[2pt]
(\ts A_{\ts1k}+B_{\ts1k}\,\si_{\ts1k}\ts)\,\,
\Bigl(\ 
\prod_{k<\ts l\ts\le N}A_{\ts1l}\,+\!
\sum_{k<j\le N}B_{\ts1j}\,
\Bigl(\ \prod_{\substack{k<\ts l\ts\le N\\l\neq j}}A_{\ts jl}\ts
\Bigr)\,\si_{\ts1j}\ts
\Bigr)\,\ga_{\ts1}\,.
\end{gather*}
In its turn, the expression in the last displayed line can be rewritten as
\begin{gather*}
\Bigl(\ 
\prod_{k\le\ts l\ts\le N}\!A_{\ts1l}\,+\,
B_{\ts1k}\,
\Bigl(\ \prod_{k<\ts l\ts\le N}A_{\ts kl}\ts
\Bigr)\,\si_{\ts1k}\ +
\\[2pt]
\sum_{k<j\le N}\,
(\ts A_{\ts1k}\,B_{\ts1j}+B_{\,1k}\,B_{\ts kj}\,\si_{\ts1k}\ts)\,\,
\Bigl(\ \prod_{\substack{k<\ts l\ts\le N\\l\neq j}}A_{\ts jl}\ts
\Bigr)\,\si_{\ts1j}\ts
\Bigr)\,\ga_{\ts1}\,.
\end{gather*}
Here none of the indices of the factor $A_{\ts jl}$ can be equal to
$1$ or $k\,$, because $j>k$ and $l\ts>k\,$. 
Further, here 
$\si_{\ts1k}\,\si_{\ts1j}=\si_{\ts1j}\,\si_{jk}$
where the factor $\si_{jk}$ commutes with the operator $\ga_{\ts1}$ 
on 
$\FF[\ts x_1\lcd x_N\ts]$
and acts trivially on the subspace \eqref{km}. Thus by the identity
$$
A_{\ts1k}\,B_{\ts1j}+B_{\,1k}\,B_{\ts kj}=B_{\ts1j}\,A_{\ts jk}
$$
the action of the operator $C_1$ on $\La_N^{\ts(k-1)}$ coincides with
the action of
\begin{gather*}
R_{\ts 12}\ldots R_{\ts 1,k-1}\ \times
\\[6pt]
\Bigl(\ 
\prod_{k\le\ts l\ts\le N}\!A_{\ts1l}\,+\,
B_{\ts1k}\,
\Bigl(\ \prod_{k<\ts l\ts\le N}A_{\ts kl}\ts
\Bigr)\,\si_{\ts1k}\,+\!
\sum_{k<j\le N}B_{\ts1j}\,
\Bigl(\ \prod_{\substack{k\le\ts l\le\ts N\\l\neq j}}A_{\ts jl}\ts
\Bigr)\,\si_{\ts1j}\ts
\Bigr)\,\ga_{\ts1}
\\
\ =
R_{\ts 12}\ldots R_{\ts 1,k-1}\,\,
\Bigl(\ 
\prod_{k-1<\ts l\ts\le N}A_{\ts1l}\,+\!
\sum_{k-1<j\le N}B_{\ts1j}\,
\Bigl(\ \prod_{\substack{k-1<\ts l\ts\le N\\l\neq j}}A_{\ts jl}\ts
\Bigr)\,\si_{\ts1j}\ts
\Bigr)\,\ga_{\ts1}\,.
\end{gather*}
Thus we have made the downward induction step.
\qed
\end{proof}

We have already noted that for any $i=1\lcd N$ the 
Cherednik operator $C_i$ maps the
polynomials in the variables $x_1\lcd x_N$ to polynomials.
On the other hand, the operator $Z_i$ commutes with those permutations
of the variables that preserve $x_i\,$.
By using these two observations when $i=1\ts$, our proposition implies

\begin{cor}
Both operators $C_1$ and $Z_1$ map the space $\La_N^{\ts(1)}$ to itself.
\end{cor}


\subsection{Main identity}
\label{sec:16}

Our main result of the current section is the theorem below. 
Define operators $U_1\lcd U_N$ acting 
on $\FF(\ts x_1\lcd x_N)$ by setting
\begin{equation}
\label{ui}
U_i\,=\,(\,t-1\ts)\,\ts
\Bigl(\,\,\ts
\prod_{l\neq i}\,A_{\ts il}\ts
\Bigr)\,\ga_{\ts i}\,.
\end{equation}
Similarly to $Z_1\lcd Z_N$ the operators $U_1\lcd U_N$
make a covariant set relative to the action of the
group $\Sg_N$ by permutations of the variables $x_1\lcd x_N\,$:
\begin{equation*}
\label{cou}
\si^{\ts-1}\ts U_i\,\si=U_{\si(i)}
\quad\text{for}\quad
\si\in\Sg_N\,.
\end{equation*}

\begin{teo}
The action of the ratio 
$D_N(\ts u\,t\ts)\ts/D_N(u)$ on $\La_N$
coincides with the action of the sum
\begin{equation}
\label{uzi}
1\,+\,u\,\sum_{i=1}^N\,U_i\,(\ts1+u\,Z_i\ts)^{\ts-1}\,.
\end{equation}
\end{teo}

\begin{proof}
We will relate operators on the space $\FF(\ts x_1\lcd x_N)$ 
by the symbol $\sim$ if their actions on the subspace $\La_N$ coincide. 
In Subsection \ref{sec:12} we already noted that the coefficients of the
polynomial $D_N(u)$ map the space $\La_N$ to itself. Let us multiply
by $D_N(u)$ on the right both the ratio and the sum appearing in our theorem,
and then subtract $D_N(u)$ from the results. We get to prove the relation
\begin{equation}
\label{dif}
D_N(u\,t\ts)-D_N(u)
\,\sim\,u\,
\sum_{i=1}^N\,U_i\,(1+u\,Z_i)^{\ts-1}\,D_N(u)\,.
\end{equation}

In the notation of Subsection \ref{sec:12} the left hand
side of the relation \eqref{dif}~equals 
$$
u\,\sum_{k=1}^N\,u^{\ts k-1}(\,t^{\,k}-1\ts)
\sum_{|I|=k}S_I(x_1\lcd x_N)\,\prod_{i\in I}\ts\ga_{\ts i}\,.
$$
Now consider the summand at the right hand side of \eqref{dif}
with the index $i=1\,$. By Proposition \ref{sec:13}
the action of this summand on $\La_N$ coincides with that of 
$$
U_1\,(1+u\,Z_1)^{-1}\,
\prod_{j=1}^N\,\,(1+u\,C_j)
\,\sim\,
U_1\,
\prod_{j=2}^N\,\,(1+u\,C_j)
\,\sim\,
U_1
\prod_{j=1}^{N-1}\ts(1+u\,t\,C_j^{\ts(1)}\ts)
$$
where we used Proposition \ref{sec:15} and then
Lemma \ref{sec:14} in the particular case $k=1\,$.
Hence by applying Proposition \ref{sec:13} once again, 
but to the Cherednik operators 
$$
C_1^{\ts(1)}\lcd\ts C_{N-1}^{\ts(1)}
$$
instead of $C_1\lcd C_N$ we obtain that 
the summand at the right hand side of the relation \eqref{dif}
with the index $i=1$ acts on $\La_N$ as
\begin{gather}
\nonumber
U_1\,\sum_{k=1}^N\,(\ts u\,t\ts)^{\ts k-1}
\!\!\sum_{|J|=k-1}\!\!
S_J(x_2\lcd x_N)\,\prod_{j\in J}\ts\ga_{\ts j}\,=
\\
\label{gaj}
\sum_{k=1}^N\,(\ts u\,t\ts)^{\ts k-1}\ts(\,t-1\ts)
\!\sum_{|J|=k-1}\!
\Bigl(\ 
\prod_{l\neq 1}\,A_{\ts1l}\,
\Bigr)\,
S_J(x_2\lcd x_N)\,\ts
\ga_{\ts 1}\ts\prod_{j\in J}\ts\ga_{\ts j}\,.
\end{gather}
Here $J$ ranges over all subsets of $\{2\ts\lcd N\}$ of size $k-1\,$.
It follows that the summands at the right hand side of \eqref{dif}
with 
$i=2\ts\lcd N$ act on $\La_N$  
as the operators obtained from \eqref{gaj} via conjugation by 
$\si_{\ts12}\,\lcd\si_{\ts1N}$ respectively.

Thus the right hand side of \eqref{dif}
acts on $\La_N$ as the operator sum of the~form
$$
\sum_{k=1}^N\,u^{\ts k-1}
\sum_{|I|=k}T_{\ts I\ts}(x_1\lcd x_N)\,\prod_{i\in I}\ts\ga_{\ts i}
$$
where $I$ ranges over all subsets of $\{1\lcd N\}$ of size $k\,$,
and each $T_{\ts I\ts}(x_1\lcd x_N)$ is a certain rational function of
the variables $x_1\lcd x_N$ over the field $\QQ\ts(t)\,$.
To prove the relation \eqref{dif} it now suffices to demonstrate that
for each $I$ 
\begin{equation}
\label{ver}
(\,t^{\,k}-1\ts)\,S_I(x_1\lcd x_N)=
T_{\ts I\ts}(x_1\lcd x_N)\,.
\end{equation}
Moreover, because both sides of \eqref{dif}
are invariant under conjugation by the elements of $\Sg_N\ts$,
it suffices to verify \eqref{ver} only in the case when $I=\{1\lcd k\}\,$.
Note that in the latter case the left hand side of \eqref{ver} equals  
\begin{equation}
\label{idone}
(\,t^{\,k}-1\ts)\,\,t^{\,k\ts(k-1)/2}
\!\!\prod_{\substack{1\le\ts j\ts\le k\\k<\ts l\ts\le N}}\!\!A_{\ts jl}\,.
\end{equation}

Now consider the right hand side of \eqref{ver} in the case when
$I=\{1\lcd k\}\,$. Let us denote it by $T$ for short. 
The contribution to $T$ from \eqref{gaj}
corresponds to the set $J=\{2\ts\lcd k\}$ and hence equals
\begin{align}
\nonumber
t^{\,k-1}\ts(\,t-1\ts)\,
\Bigl(\ 
\prod_{l\neq 1}\,A_{\ts 1l}\,
\Bigr)\,\ts
&t^{\,(k-1)\ts(k-2)/2}
\!\!\prod_{\substack{2\le\ts j\ts\le k\\k<\ts l\ts\le N}}\!\!A_{\ts jl}\ =
\\[2pt]
\label{idtwo}
(\,t-1\ts)\,
\Bigl(\ 
\prod_{1<\ts l\ts\le k}\,A_{\ts 1l}\,
\Bigr)\,\ts
&t^{\,k\ts(k-1)/2}
\!\!\prod_{\substack{1\le\ts j\ts\le k\\k<\ts l\ts\le N}}\!\!A_{\ts jl}\,.
\end{align}
If we conjugate \eqref{gaj} by any $\si_{\ts1i}$ with $i>1\,$,
the result will make contribution to $B$ only when $i\le k\,$,
and this contribution will correspond to $J=\{2\ts\lcd k\}\,$. 
Indeed, then we will need $J$ in \eqref{gaj} such that 
$\si_{\ts1i}\,(\{1\}\sqcup J)=\{1\lcd k\}\,$. 
Hence for each index $i=2\ts\lcd k$ we get a contribution to $T$
\begin{align}
\nonumber
\si_{\ts1i}\ts\Bigl(
(\,t-1\ts)\,
\Bigl(\ 
\prod_{1<\ts l\ts\le k}\,A_{\ts 1l}\,
\Bigr)\,\ts
&t^{\,k\ts(k-1)/2}
\!\!\prod_{\substack{1\le\ts j\ts\le k\\k<\ts l\ts\le N}}\!\!A_{\ts jl}\ts
\Bigr)=
\\[2pt]
\label{three}
(\,t-1\ts)\,
\Bigl(\ 
\prod_{\substack{1\le\ts l\ts\le k\\l\neq i}}\,A_{\ts il}\,
\Bigr)\,\ts
&t^{\,k\ts(k-1)/2}
\!\!\prod_{\substack{1\le\ts j\ts\le k\\k<\ts l\ts\le N}}\!\!A_{\ts jl}\,.
\end{align}

By dividing \eqref{idone},\eqref{idtwo},\eqref{three} by 
$(\,t-1\ts)\,t^{\,k\ts(k-1)/2}$
and by cancelling there~all the common factors $A_{\ts jl}$ 
the relation \eqref{ver} now reduces to the identity
$$
\frac{\,t^{\,k}-1}{t-1}\,=\,\sum_{i=1}^k\ 
\prod_{\substack{1\le j\le k\\j\neq i}}\,
\frac{\,x_i-t\,x_j}{x_i-x_j}\ .
$$
The latter identity is easy to verify, and we omit the details
of verification.
\qed
\end{proof}

Consider the operator sum over $i=1\lcd N$ appearing in
\eqref{uzi}. Denote by $I_N(u)$ the restriction of this operator sum
to the subspace $\La_N\subset\FF\ts(\ts x_1\lcd x_N)\,$.  
The $I_N(u)$ expands as a formal power series in 
$u$ with coefficients acting on~$\La_N\,$.
Our theorem means that the action of the coefficients of the series
$1+u\,I_N(u)$ on $\La_N$ coincides with the action of the respective 
coefficients~of~$D_N(\ts u\,t\ts)\ts/D_N(u)\,$.
The latter ratio should be also  
expanded as a formal power series in $u$ here.  

The coefficients of the series $I_N(u)$ 
will be called the \textit{quantum Hamiltonians} corresponding to the 
basis of Macdonald polynomials in the vector space $\La_N\,$. 
In the next subsection we give another expression for 
$I_N(u)$ by using the resolvent of a certain $N\times N$ matrix
with operator entries which act on $\FF\ts(\ts x_1\lcd x_N)\,$.


\subsection{Matrix resolvent}
\label{sec:17}

Take any $f\in\La_N^{\ts(1)}$ and consider 
the column vector
$$
\Fc\,=\,
\begin{bmatrix}
f
\\
\si_{\ts12}\,(f)
\\
\vdots
\\
\ \si_{\ts1N}\ts(f)\ 
\end{bmatrix}
$$
Now define a $N\times N$ matrix $\Zc$ with operator entries 
acting on the vector space $\FF\ts(\ts x_1\lcd x_N)$ as follows.
The $i\,,j\ts$-entry $Z_{\ts ij}$ of the matrix $\Zc$ is defined by setting 
\begin{gather}
Z_{\ts ii}=
\Bigl(\,\,\ts
\prod_{l\neq i}\,A_{\ts il}\ts
\Bigr)\,\ga_{\ts i}\,\ts,
\\
Z_{\ts ij}=
B_{\ts ij}\,
\Bigl(\,\prod_{l\neq i,j}A_{\ts jl}\ts\Bigr)\,\ga_{\ts j}
\quad\text{for}\quad i\neq j\,.
\end{gather}
Then by using the definition \eqref{wi} we have
\begin{equation}
\label{Zij}
Z_i\,=\,W_i\ts\ga_{\ts i}\,=\,Z_{\ts ii}\,+\,
\sum_{j\neq i}\,Z_{\ts ij}\,\si_{ij}
\end{equation}
where for $j\neq i$ we also use the relation 
$\si_{ij}\,\ga_{\ts i}=\ga_{\ts j}\,\si_{ij}\,$.
It follows from \eqref{Zij}~that
\begin{equation}
\label{ZF}
\begin{bmatrix}
Z_1\ts(f)
\\[2pt]
Z_{\ts2}\,\ts\si_{\ts12}\,(f)
\\
\vdots
\\[2pt]
\ \ts
Z_N\,\si_{\ts1N}\ts(f)
\ 
\end{bmatrix}
\,=\,\Zc\,\Fc\,.
\end{equation}
Indeed, in its first entry the vector
equality \eqref{ZF} holds by \eqref{Zij} with $i=1\,$.
If $i\neq1$ then by using \eqref{Zij} we get
\begin{gather*}
Z_i\,\si_{\ts1i}\,(f)\,=\,
Z_{\ts ii}\,\si_{\ts1i}\,(f)+
Z_{\ts i1}\,\si_{\ts i1}\,\si_{\ts1i}\,(f)+
\sum_{j\neq1,i}\,Z_{\ts ij}\,\si_{ij}\,\si_{\ts1i}\,(f)\,=
\\
Z_{\ts ii}\,\si_{\ts1i}\,(f)+
Z_{\ts i1}\ts(f)+
\sum_{j\neq1,i}\,Z_{\ts ij}\,\si_{\ts1j}\,(f)\,=\,
Z_{\ts i1}\ts(f)+
\sum_{j\neq1}\,Z_{\ts ij}\,\si_{\ts1j}\,(f)
\end{gather*}
as required. Here for three pairwise distinct indices $1\ts,i\ts,j$
we also use the relations
$$
\si_{ij}\,\si_{\ts1i}\,(f)=
\si_{\ts1j}\,\si_{ij}\,(f)=
\si_{\ts 1j}\,(f)\,.
$$

By the the covariance property \eqref{coz} 
of the operators $Z_1\lcd Z_N$ the column vector at the left hand
side of \eqref{ZF} has a form similar to $\Fc\,$. Namely, it can be
obtained by replacing the polynomial $f$ in $\Fc$ by $Z_1(f)\,$.
Here we use Corollary~\ref{sec:15}.
By expanding 
$(\ts1+u\,Z_i\ts)^{\ts-1}$ for every $i=1\lcd N$
as a formal power series in $u$ 
and by repeatedly using the above arguments,
we get the equality
\begin{equation}
\label{ZZZ}
\begin{bmatrix}
(\,1+u\,Z_1\ts)^{\ts-1}\ts(f)
\\[4pt]
(\,1+u\,Z_{\ts2}\ts)^{\ts-1}\,\si_{\ts12}\,(f)
\\
\vdots
\\[4pt]
\ (\,1+u\,Z_N\ts)^{\ts-1}\,\si_{\ts1N}\ts(f)\ 
\end{bmatrix}
\,=\,(\,1+u\,\Zc\,)^{\ts-1}\,\Fc\,.
\end{equation}

Now suppose $f\in\La_N$ so that
the polynomial $f$ is symmetric in all the variables
$x_1\lcd x_N\ts$. Then 
$$
f=\si_{\ts12}\,(f)=\ldots=\si_{\ts1N}\ts(f)
$$
so that $\Fc=\Ec\ts f$ where $\Ec$ is 
the column vector of size $N$ with every entry being~$1\,$.
Let $\Uc$ be the row vector of size $N$ where
the $i\ts$-entry is the $U_i$ defined by~\eqref{ui}. 
By using \eqref{ZZZ} and the definition of the series $I_N(u)$ 
as given 
in Subsection \ref{sec:16} 
$$
I_N(u)\,f=\Uc\,\ts(\,1+u\,\Zc\,)^{\ts-1}\,\Ec\,f\,.
$$
Thus we have proved that the action of $I_N(u)$ on $\La_N$
coincides with the action~of 
\begin{equation}
\label{Ham}
\Uc\,\ts(\,1+u\,\Zc\,)^{\ts-1}\,\Ec\,.
\end{equation}
Hence we now obtain the following corollary to Theorem \ref{sec:16}.

\begin{cor}
The action of the ratio 
$D_N(\ts u\,t\ts)\ts/D_N(u)$ on $\La_N$
coincides with that~of
\begin{equation*}
1\,+\,u\,\,\Uc\,\ts(\,1+u\,\Zc\,)^{\ts-1}\,\Ec\,.
\end{equation*}
\end{cor}


\section{Inverse limits}
\label{sec:3}


\subsection{Limits of covariant operators}
\label{sec:31}

Let $\FF=\QQ\ts(q,t)$ as before. 
We will find first the inverse limit at $N\to\infty$
of the restriction of the operator $Z_1$ to the subspace
\begin{equation}
\label{subspace}
\La_N^{\ts(1)}\subset\ts\FF\ts(\ts x_1\lcd x_N)\,.
\end{equation}
By Proposition \ref{sec:15} the 
operator $C_1$
has the same restriction to $\La_N^{\ts(1)}$.
The limit will be an operator acting on the space $\La\ts[\ts v\ts]$
and denoted simply by $Z\,$.
To define the limit extend the canonical homomorphism $\La\to\La_N$
to a homomorphism
\begin{equation*}
\label{piv} 
\pi_N:\,\La\ts[\ts v\ts]\,\to\,\La_N^{\ts(1)}:\,
v\,\mapsto\,x_1\,.
\end{equation*}
Here
$$
\pi_N:\,p_n\,\mapsto\,p_n(\ts x_1\lcd x_N)
\quad\text{for}\quad
n=1\ts,2\ts,\,\ldots\,.
$$

We will now define an operator $Z$ on the vector space 
$\La\ts[\ts v\ts]$ explicitly. Denote by $\xi$ and $\eta$
the automorphisms of the $\FF$-algebra $\La\ts[\ts v\ts]$
which act trivially on the subalgebra $\La$ but map the variable 
$v$ to $\,q^{\ts-1}\ts v\,$ and $\,t\,v\,$ respectively.
Thus $\xi$ is the inverse $q\ts$-shift of $v$
while $\eta$ is the usual $t\ts$-shift.
Next equip the vector space $\FF\ts[\ts v\ts]$ with the standard inner product
so that $1\ts,v\ts,v^{\ts2},\ts\ldots$ form an orthonormal basis.
Denote by $v^{\ts\circ}$ the operator
on $\FF\ts[\ts v\ts]$ conjugate to multiplication by $v\,$.~Explicitly, 
\begin{equation*}
v^{\ts\circ}:\,v^n\,\mapsto\, 
\left\{
\begin{array}{cl}
v^{\ts n-1}
&\quad\textrm{if}\quad\,n>0\,,
\\[2pt]
0
&\quad\textrm{if}\quad\,n=0\,.
\end{array}
\right.
\end{equation*}

Extend the operator $v^{\ts\circ}$ from $\FF\ts[\ts v\ts]$ to 
$\La\ts[\ts v\ts]$ by $\La\ts$-linearity.
The extension will still be denoted by $v^{\ts\circ}$. 
For every $f\in\La$ 
extend from $\La$ to $\La\ts[\ts v\ts]$ by
$\FF\ts[\ts v\ts]\ts$-linearity the operator
of multiplication by $f$ and its conjugate 
operator $f^{\ts\perp}\,$.
Recall that the superscript ${}^\perp$ here indicates
conjugation relative to the inner product \eqref{schurprod}.
The conjugate $f^{\ts\ast}$ relative to the inner product
\eqref{macprod} extends from $\La$ to $\La\ts[\ts v\ts]$ 
in the same way as $f^{\ts\perp}$ does.
Using these conventions, put $Z=W\ts\ga$ where
\begin{equation*}
\label{Wga}
\ga=\xi\ts\,Q^{\ts\ast}(v)
\quad\text{and}\quad
W=\eta\ts\,Q\ts(\ts v^{\ts\circ}\ts)\,.
\end{equation*}

\begin{teo}
We have a commutative diagram of\/ $\FF$-linear mappings
\begin{equation} 
\label{cd1}    
\begin{tikzcd}[row sep=32pt,column sep=32pt]
\La\ts[\ts v\ts]
\arrow[xshift=0pt,yshift=0pt]{r}[above=1pt]{Z}        
\arrow[xshift=0pt,yshift=0pt]{d}[left=1pt]{\pi_N} 
&
\La\ts[\ts v\ts]
\arrow[xshift=0pt,yshift=0pt]{d}[right=2pt]{\pi_N}
\\
\La_N^{\ts(1)}
\arrow[xshift=0pt,yshift=0pt]{r}[below=2pt]{Z_1}   
&
\La_N^{\ts(1)} 
\end{tikzcd}
\vspace{2pt}  
\end{equation}
\end{teo}  

\begin{proof} 
We will verify commutativity of the two diagrams obtained from
\eqref{cd1} by replacing $Z\ts,Z_1$ respectively by
$\ga\ts,\ga_{\ts1}$ and $W\ts,W_1\,$. Our theorem will then follow.

Firstly observe that the extended operator 
$Q^{\ts\ast}(v):\La\ts[\ts v\ts]\to\La\ts[\ts v\ts]$ 
appearing in the definition of $\ga$ is a homomorphism
of $\FF$-algebras, and so is $\xi:\La\ts[\ts v\ts]\to \La\ts[\ts v\ts]\,$.
Hence it suffices to show that the compositions
$\pi_N\,\ga$ and $\ga_{\ts1}\,\pi_N$ coincide on $v$
and on $p_n$ for $n=1\ts,2\ts,\,\ldots\ $.
By applying the compositions to $v$ we get the~same~result:
$$   
\begin{tikzcd}[column sep=32pt]
v
\arrow[mapsto]{r}[below=2pt]{\ga}        
&
q^{\ts-1}\ts v
\arrow[mapsto]{r}[below=2pt]{\pi_N}
& 
q^{\ts-1}\ts x_1
\end{tikzcd}
\quad\text{and}\quad   
\begin{tikzcd}[column sep=32pt]
v
\arrow[mapsto]{r}[below=2pt]{\pi_N}        
&
x_1
\arrow[mapsto]{r}[below=2pt]{\ga_{\ts1}} 
&
q^{\ts-1}\ts x_1\,.
\end{tikzcd}
$$

To check the coincidence on any generator $p_n$ note that by 
the identity~\eqref{qhh}
$$
\ga=
\xi\ts\,H^{\ts\perp}(\ts v\,q\ts)^{\ts-1}\,H^{\ts\perp}(v)\,=
H^{\ts\perp}(\ts v\ts)^{\ts-1}\,\xi\ts\,H^{\ts\perp}(v)\,.
$$
Hence due to \eqref{vop} by applying
$\pi_N\,\ga$ and $\ga_{\ts1}\,\pi_N$ to $p_n$ we also get the same result:
$$   
\begin{tikzcd}[column sep=32pt]
p_n
\arrow[mapsto]{r}[below=2pt]{\ga}        
& 
q^{\ts-n}\ts v^n-v^n+p_n
\arrow[mapsto]{r}[below=2pt]{\pi_N} 
&  
q^{\ts-n}\ts x_1+x_2^n+\ldots+x_N^n
\end{tikzcd}
\vspace{-8pt}
$$
and
$$   
\begin{tikzcd}[column sep=32pt]
p_n
\arrow[mapsto]{r}[below=2pt]{\pi_N}        
&
x_1^n+\ldots+x_N^n
\arrow[mapsto]{r}[below=2pt]{\ga_{\ts1}} 
&
q^{\ts-n}\ts x_1+x_2^n+\ldots+x_N^n\,.
\end{tikzcd}
\vspace{2pt}
$$

Now consider the compositions
$\pi_N\ts W$ and $W_1\ts\pi_N\,$.
By definition, the extended operator 
$W:\La\ts[\ts v\ts]\to \La\ts[\ts v\ts]$
commutes with multiplication by any $f\in\La\,$. But
\begin{equation}
\label{W1}
W_1\,=\, 
\prod_{1<\ts l\ts\le N}\,A_{\ts 1l}\,+\,
\sum_{1<j\le N}\,B_{\ts1j}\,
\Bigl(\,
\prod_{\substack{1<\ts l\ts\le N\\l\neq j}}\,A_{\ts jl}
\ts\Bigr)\,\si_{1j}
\end{equation}
by 
\eqref{wi}. In particular, 
the restriction of the operator $W_1$ to the subspace \eqref{subspace}
commutes with multiplication by $\pi_N(f)\in\La_N\,$.
Hence it suffices to show that the compositions
$\pi_N\ts W$ and $W_1\ts\pi_N$ coincide on the elements
$1\ts,v\ts,v^2\ts,\,\ldots\ \in\La\ts[\ts v\ts]\,$.
Let us use the generating series of these elements
\begin{equation}
\label{uv}
1+u\,v+u^2\ts v^2+\ldots\,=\,\frac1{\,1-u\,v\,}
\end{equation}
in the other variable $u\,$. 
By applying $\pi_N\ts W$ to the series \eqref{uv} we get
\begin{equation}
\label{uv1}   
\begin{tikzcd}[column sep=32pt]
\displaystyle
\frac1{\,1-u\,v\,}
\arrow[mapsto]{r}[below=2pt]{{\small W}}        
&
\displaystyle
\frac{\,Q(u)}{\,1-u\,t\,v\,}
\arrow[mapsto]{r}[below=2pt]{\pi_N} 
&
\displaystyle
\frac1{\,1-u\,t\,x_1}
\,\,\prod_{i=1}^N\,\ts
\frac{\,1-u\,t\,x_i\,}{1-u\,x_i}\ .
\end{tikzcd}
\end{equation}
Here we employed the general fact that 
for any formal power series $G\ts(u)$ with~the coefficients from $\FF$
$$
G\ts(\ts v^{\ts\circ}\ts)\,\frac1{\,1-u\,v\,}=
\frac{\,G\ts(u)}{\,1-u\,v\,}\ .
$$  
We have also employed the relation \eqref{qvprod} with the variable $v$
replaced by $u\,$. 

On the other hand, by applying $W_1\ts\pi_N$ to the series 
\eqref{uv} we get
$$  
\begin{tikzcd}[column sep=32pt]
\displaystyle
\frac1{\,1-u\,v\,}
\arrow[mapsto]{r}[below=2pt]{{\pi_N}}        
&
\displaystyle
\frac{1}{\,1-u\,x_1\,}
\arrow[mapsto]{r}[below=2pt]{W_1} 
&
\displaystyle
\frac1{\,1-u\,x_1}
\prod_{1<\ts l\ts\le N}\frac{\,x_1-t\,x_l}{x_1-x_l}
\ \,+
\end{tikzcd}
\vspace{-2pt}
$$
\begin{equation*}
\label{uv2}
\sum_{1<j\le N}\,
\frac{(\,t-1\ts)\,x_j}{(\,1-u\,x_j\ts)\,(\,x_1-x_j\ts)}\,
\prod_{\substack{1<l\le N\\l\neq j}}\,
\frac{\,x_j-t\,x_l}{x_j-x_l}\ .
\vspace{2pt}
\end{equation*}
Here we also used \eqref{AB} and \eqref{W1}.
It easy to verify that the results 
obtained in \eqref{uv1} and in the last two displayed
lines are the same. Consider them
as rational functions of 
$u$ and assume that 
$x_1\lcd x_N\neq 0\,$. Then
both rational functions vanish at $u=\infty$
and have poles only at $u=x_1^{\ts-1}\lcd x_N^{\ts-1}\,$.
All these poles are simple, and the corresponding residues 
of the two functions coincide.
\qed  
\end{proof}

Note that for any index $i=2\lcd N$ one can also consider the restriction  
of the operator $Z_i$ to the subspace of $\FF(\ts x_1\lcd x_N)$
consisting of the polynomials in $x_1\lcd x_N$ symmetric in all the variables
but $x_i\,$. By the covariance property 
\eqref{coz} our Corollary \ref{sec:15} implies that the operator $Z_i$ 
preserves this subspace. We could have defined
the extension $\pi_N$ of the homomorphism $\La\to\La_N$
from $\La$ to $\La\ts[\ts v\ts]$ by mapping the variable $v$
to $x_i$ instead of $x_1\,$. The image of $\pi_N$ would be then 
the latter subspace of $\ts\FF\ts(\ts x_1\lcd x_N)\,$.
The inverse limit of the restriction of $Z_i$ to that subspace
would be then the same operator $Z$ acting on $\La\ts[\ts v\ts]\ts$.
This coincidence follows immediately from the property \eqref{coz}.

It is the change of parameters $q\mapsto q^{\ts-1}$ and
$t\mapsto t^{\ts-1}$ in the original definition [VI.3.2] that
allowed us to state the last theorem in terms
of the Hall-Littlewood symmetric functions
$Q_1\ts,Q_{\ts2}\ts,\ts\ldots\ $.
Otherwise we would have to change
$t\mapsto t^{\ts-1}$ in the definition of the latter 
symmetric functions. 
The change of the variable $X$ in 
[VI.3.2] and the corresponding choice of normalization of
the operator $C_i$ as in \eqref{ci} and as in Proposition \ref{sec:13} 
ensure that every $Z_i$ has a limit at $N\to\infty\,$.


\subsection{Limits of quantum Hamiltonians}
\label{sec:32}

In this subsection
we will find the inverse limits at $N\to\infty$
of the quantum Hamiltonians
corresponding to the basis of Macdonald polynomials 
in $\La_N\,$. These quantum Hamiltonians are
defined as the operator coefficients of the series $I_N(u)$  
acting on the vector space 
$\La_N\,$, see the end of Subsection~\ref{sec:16}.
We will denote by $I\ts(u)$ the inverse limit of the series~$I_N(u)\ts$.

The coefficients of the series $I\ts(u)$
will be certain operators $\La\to\La\ts[\ts w\ts]$
where $w$ is yet another formal variable.
We will then eliminate the dependence of the
coefficients on $w$ by renormalising the series $I\ts(u)\,$.
Hence the coefficients of the renormalised series \eqref{ju}
will be operators acting on $\La\,$. 

Consider the sum \eqref{uzi} over $i=1\lcd N$ appearing in \eqref{uzi}.
By 
\eqref{ui} the action of this sum
on the subspace $\La_N\subset\FF\ts(\ts x_1\lcd x_N)$ 
coincides with that of
\begin{equation}
\label{uz1}
V_{\ts1}\,\ga_1\,(\,1+u\,Z_1\ts)^{\ts-1}
\end{equation}
where we set
$$
V_{\ts1}\,=\,
(\,t-1\ts)\,\,\sum_{i=1}^N\,\,
\Bigl(\,\,
\prod_{\substack{1\le\ts l\ts\le N\\l\neq i}}\,A_{\ts il}\ts
\Bigr)\,\ts
\si_{\ts1i}\,\ts.
$$
Here $\si_{\ts11}=1\,$. We will demonstrate
that the operator $V_{\ts1}$ maps the subspace \eqref{subspace} to~$\La_N\,$. 
At the same time we will determine the 
inverse limit at $N\to\infty$ of the restriction
of the operator $V_{\ts1}$ to the subspace \eqref{subspace}.
The latter limit will be an operator
$\La\ts[\ts v\ts]\to\La\ts[\ts w\ts]$ denoted simply by $V$.
To determine this limit
extend the canonical homomorphism $\La\to\La_N$ to a homomorphism
\begin{equation*}
\label{tau} 
\tau_{\ts N}:\,\La\ts[\ts w\ts]\to\La_N:\,
w\,\mapsto\,t^{\ts N}\,.
\end{equation*}
Here
$$
\tau_{\ts N}:\,p_n\,\mapsto\,p_n(\ts x_1\lcd x_N)
\quad\text{for}\quad
n=1\ts,2\ts,\,\ldots\,.
$$
This definition of the homomorphism $\tau_N$ 
goes back to \cite[Section 6]{Rains}.
Now define $V$ explicitly as the unique $\La\ts$-linear operator
$\La\ts[\ts v\ts]\to\La\ts[\ts w\ts]$ such that  
\begin{equation*}
V:\,v^n\,\mapsto\, 
\left\{
\begin{array}{cl}
-\,Q_n
&\quad\textrm{if}\quad\,n>0\,,
\\[2pt]
w-1
&\quad\textrm{if}\quad\,n=0\,.
\end{array}
\right.
\end{equation*} 
 
\begin{pro}
We have a commutative diagram of\/ $\FF$-linear mappings
\begin{equation} 
\label{cd2}    
\begin{tikzcd}[row sep=32pt,column sep=32pt]
\La\ts[\ts v\ts]
\arrow[xshift=0pt,yshift=0pt]{r}[above=1pt]{V}        
\arrow[xshift=0pt,yshift=0pt]{d}[left=1pt]{\pi_N} 
&
\La\ts[\ts w\ts]
\arrow[xshift=0pt,yshift=0pt]{d}[right=2pt]{\tau_{N}}
\\
\La_N^{\ts(1)}
\arrow[xshift=0pt,yshift=0pt]{r}[below=2pt]{V_1}   
&
\,\La_N^{\phantom{(1)}}
\end{tikzcd}
\vspace{2pt}  
\end{equation}
\end{pro}   

\begin{proof}
The operator 
$V:\La\ts[\ts v\ts]\to\La\ts[\ts w\ts]$
commutes with the multiplication by any $f\in\La\,$.
In turn, the restriction of the operator $V_1$ 
to the subspace \eqref{subspace}~commutes 
with multiplication by $\pi_N(f)\in\La_N\,$.
So it suffices to show that the compositions
$\tau_{\ts N}\,V$ and $V_1\,\pi_N$ coincide on the elements
$1\ts,v\ts,v^2\ts,\,\ldots\ \in\La\ts[\ts v\ts]\,$.
Let us again use the generating series \eqref{uv}
of these elements. By applying $\tau_{\ts N}\,V$ to \eqref{uv} we get
\begin{equation}
\label{uv3}   
\begin{tikzcd}[column sep=32pt]
\displaystyle
\frac1{\,1-u\,v\,}
\arrow[mapsto]{r}[below=2pt]{{\small V}}        
&
w-Q(u)
\arrow[mapsto]{r}[below=2pt]{\tau_{\ts N}} 
&
\displaystyle
t^{\ts N}-\ts\,\prod_{i=1}^N\,\ts
\frac{\,1-u\,t\,x_i\,}{1-u\,x_i}
\end{tikzcd}
\end{equation}
where we used 
\eqref{qvprod}.
On the other hand, by applying $V_1\,\pi_N$ to 
\eqref{uv} we get
\begin{equation}
\label{uv4}  
\begin{tikzcd}[column sep=32pt]
\displaystyle
\frac1{\,1-u\,v\,}
\arrow[mapsto]{r}[below=2pt]{{\pi_N}}        
&
\displaystyle
\frac{1}{\,1-u\,x_1\,}
\arrow[mapsto]{r}[below=2pt]{V_1} 
&
\displaystyle
\,\,\sum_{i=1}^N\,\,
\frac{t-1}{\,1-u\,x_i}
\,\prod_{\substack{1\le\ts l\ts\le N\\l\neq i}}\frac{\,x_i-t\,x_l}{x_i-x_l}
\ .
\end{tikzcd}
\end{equation}
The results 
obtained in \eqref{uv3} and \eqref{uv4} are equal to each other.
Indeed, consider them
as rational functions of 
$u$ and assume that $x_1\lcd x_N\neq 0\,$. Then
both rational functions vanish at $u=\infty$
and have poles 
at $u=x_1^{\ts-1}\lcd x_N^{\ts-1}\,$.
These poles are simple, and the corresponding residues 
of two functions coincide.
\qed
\end{proof}

By the surjectivity of $\pi_N$ the last proposition implies
that the operator $V_{\ts1}$ maps the subspace \eqref{subspace}
to~$\La_N\,$. Moreover, it implies that
the inverse limit at $N\to\infty$ 
of the restriction of the operator sum
\eqref{uz1} to~the subspace \eqref{subspace} equals
\begin{equation*}
\label{VgaZ}
V\,\ga\,(\,1+u\,Z\ts)^{\ts-1}\,=\,
\sum_{n=0}^\infty\,\ts (\ts-\ts u\ts)^{\ts n}\ts\,V\,\ga\,Z^{\ts n}\,.
\end{equation*}
By the definitions of $Z\ts,\ga$ and $V$ here for every $n\ge0$
the composition $V\,\ga\,Z^{\ts n}$
is an operator $\La\ts[\ts v\ts]\to\La\ts[\ts w\ts]\,$.
The above stated equality of the inverse limit 
follows from the commutativity of the diagram
\begin{equation*} 
\label{cd3}    
\begin{tikzcd}[row sep=32pt,column sep=32pt]
\La\ts[\ts v\ts]
\arrow[xshift=0pt,yshift=0pt]{r}[above=1pt]{Z^n}        
\arrow[xshift=0pt,yshift=0pt]{d}[left=1pt]{\pi_N} 
&
\La\ts[\ts v\ts]
\arrow[xshift=0pt,yshift=0pt]{r}[above=1pt]{\ga\ }        
\arrow[xshift=0pt,yshift=0pt]{d}[left=1pt]{\pi_N} 
&
\La\ts[\ts v\ts]
\arrow[xshift=0pt,yshift=0pt]{r}[above=1pt]{V}        
\arrow[xshift=0pt,yshift=0pt]{d}[right=1pt]{\pi_N} 
&
\La\ts[\ts w\ts]
\arrow[xshift=0pt,yshift=0pt]{d}[right=2pt]{\tau_{N}}
\\
\La_N^{\ts(1)}
\arrow[xshift=0pt,yshift=0pt]{r}[below=2pt]{Z_1^n}
&
\!\La_N^{\ts(1)}
\arrow[xshift=0pt,yshift=0pt]{r}[below=2pt]{\ga_{\ts1}}
&
\La_N^{\ts(1)}
\arrow[xshift=0pt,yshift=0pt]{r}[below=2pt]{V_1}   
&
\,\La_N^{\phantom{(1)}}
\end{tikzcd}
\vspace{2pt}  
\end{equation*}
Here we use the commutativity of \eqref{cd1},\eqref{cd2}
and that of the diagram obtained from \eqref{cd1} by replacing $Z\ts,Z_1$ 
respectively by $\ga\ts,\ga_{\ts1}\,$.
The commutativity of the diagram so obtained has been established 
as a part of our proof of Theorem~\ref{sec:31}.

Denote by $\de$ the embedding of $\La$ to $\La\ts[\ts v\ts]$
as the subspace of degree zero in $v\,$. Then we have a commutative 
diagram 
\begin{equation*} 
\label{cd4}    
\begin{tikzcd}[row sep=32pt,column sep=32pt]
\La{\phantom{[]}}\!\!
\arrow[xshift=0pt,yshift=0pt]{r}[above=2pt]{\de}        
\arrow[xshift=0pt,yshift=0pt]{d}[left=0pt]{} 
&
\La\ts[\ts v\ts]
\arrow[xshift=0pt,yshift=0pt]{d}[right=2pt]{\pi_N}
\\
\La_N^{\phantom{(1)}}\!\!\!
\arrow[xshift=0pt,yshift=0pt]{r}[below=0pt]{}   
&
\La_N^{(1)}
\end{tikzcd}  
\end{equation*}
where the left vertical arrow is the canonical projection. 
The bottom horizontal arrow is the natural embedding. It
follows that the inverse limit of $I_N(u)$ equals
\begin{equation}
\label{iu}
I\ts(u)\,=\,V\,\ga\,(\,1+u\,Z\ts)^{\ts-1}\,\de\,.
\end{equation}

By the above definition, every coefficient in the formal power 
series expansion of $I\ts(u)$ in $u$ is a certain operator
$\La\to\La\ts[\ts w\ts]\,$. 
Now consider the series
\begin{equation}
\label{ju}
(1+u)\,(\,1+u\,w\,)^{-1}\ts(\,1+u\,I\ts(u)\ts)
\end{equation}
where the summand $1$ 
in front of 
$u\,I\ts(u)$ stands for the embedding of $\La$ to $\La\ts[\ts w\ts]$
as the subspace of degree zero in $w\,$. This should cause no confusion. 
In the next subsection we will show that the series \eqref{ju} 
does not depend on $w\,$.
Hence the coefficients of this series will be operators 
mapping the vector space $\La$ to itself.

Note that by the definition of the homomorphism $\tau_{\ts N}$
and by the above given arguments, 
the series \eqref{ju} is equal to the inverse
limit at $N\to\infty$ of   
\begin{equation}
\label{juN}
(1+u)\,(\,1+u\,t^{\ts N}\,)^{\ts-1}\ts(\,1+u\,I_N(u)\ts)\,.
\end{equation}
But by using the multiplicative formula \eqref{eigenvalue},
the eigenvalue of $D_N(\ts u\,t\ts)\ts/D_N(u)$ 
on the trivial Macdonald polynomial $1\in\La_N$ 
corresponding to 
$\la=(\ts0\ts,0\ts,\,\ldots\,)$ is
$$
(1+u)^{\ts-1}\ts(\,1+u\,t^{\ts N}\,)\,.
$$
Hence our Theorem \ref{sec:16} 
implies that the eigenvalue of \eqref{juN} on 
$1\in\La_N$ equals~$1\,$.
By taking the limit at $N\to\infty$
the eigenvalue of \eqref{ju} on the trivial
Macdonald symmetric function $1\in\La$ 
also equals $1\,$. This explains the definition of \eqref{ju}.


\subsection{Truncated space}
\label{sec:33}

Here we will use the vector space decomposition 
\begin{equation}
\label{trunc}
\La\ts[\ts v\ts]=\La\oplus v\ts\La\ts[\ts v\ts]\,.
\end{equation}
The second direct summand in \eqref{trunc} will be called
the \textit{truncated space}. Relative to this decomposition
the operator $\ga$ on $\La\ts[\ts v\ts]$ is represented by the $2\times2$
matrix with operator entries
$$
\begin{bmatrix}
\ 1&\,0\ 
\\[2pt]
\ \be&\al\ 
\end{bmatrix}
$$
where $\be$ denotes the composition of the restriction of $\ga$ to
the first summand in \eqref{trunc}
with the projection to the second summand.
The map $\ga$ preserves the second summand, 
and $\al$ denotes the restriction of $\ga$ to it.
Similarly, the operator $W$ on $\La\ts[\ts v\ts]$ 
is represented by the $2\times2$ matrix with operator entries
$$
\begin{bmatrix}
\ 1&\,Y\,
\\[2pt]
\ 0&\,X\, 
\end{bmatrix}
$$
where $X$ and $Y$ respectively denote the compositions of the
restriction of $W$ to the second summand in \eqref{trunc}
with the projections to the first and to the second summands.
Note that the operator $Z=W\ts\ga$ is then represented by the product 
\begin{equation}
\label{Zmat}
\begin{bmatrix}
\ 1&\,Y\,
\\[2pt]
\ 0&\,X\, 
\end{bmatrix}
\begin{bmatrix}
\ 1&\,0\ 
\\[2pt]
\ \be&\al\ 
\end{bmatrix}
=
\begin{bmatrix}
\ 1+Y\be&\ Y\al\ 
\\[2pt]
\ X\ts\be&\,X\ts\al\ 
\end{bmatrix}
.
\end{equation}

By definition the operator $V:\La\ts[\ts v\ts]\to\La\ts[\ts w\ts]$ 
acts on the first direct summand in \eqref{trunc} as multiplication 
by $w-1\,$. The restriction of $V$ to the second direct summand
does not depend on $w\,$. Thus it maps $v\ts\La\ts[\ts v\ts]$ to $\La\,$.
Moreover, by the definitions of $V$ and $W$ 
this restriction coincides with the operator $-\ts Y\,$. 
Hence relative to 
\eqref{trunc} the operator $V$ is represented
by the row 
with operator entries
$$
\begin{bmatrix}
\,w-1\ts\,\ts,\,-\ts Y\,
\end{bmatrix}
.
$$

Finally denote $L=\al\,X\,$. This is an operator
on the truncated space $v\ts\La\ts[\ts v\ts]\,$. We shall call it the 
\textit{Lax operator} for the Macdonald symmetric functions in $\La\,$.
This terminology is justified by the following theorem. 

\begin{teo}
The series \eqref{ju} is equal to $\,(\ts1+u\ts)\,(\ts1+u+u\,J\ts(u))^{\ts-1}$ 
where
$$
J\ts(u)\,=\,Y\,(\,1+u\,L\ts)^{\ts-1}\ts\be\,.
$$
In particular, the series \eqref{ju} does not depend on the variable $w\,$. 
\end{teo}

\begin{proof}
Relative to 
\eqref{trunc} the operator $\de:\La\to\La\ts[\ts v\ts]$
is represented by the column with two operator entries
$$
\begin{bmatrix}
\,1\,
\\
0 
\end{bmatrix}
.
$$
Therefore the operator product $(\,1+u\,Z\ts)^{\ts-1}\,\de$ appearing in the 
definition \eqref{iu} of the series $I\ts(u)$ is represented 
by the first column of the $2\times2$ matrix inverse~to
$$
\begin{bmatrix}
\ 1+u+u\ts Y\be&\ u\,Y\al\ 
\\[2pt]
\ u\,X\ts\be&\,1+u\,X\ts\al\ 
\end{bmatrix}
.
$$
Here we employ the matrix representation \eqref{Zmat} of the operator $Z\,$.
To find that column we will use a well known formula
for the inverse of a $2\times2$ block matrix with 
invertible diagonal blocks, see \cite[Lemma 3.2]{B}.
The block matrix is assumed to be invertible too.
The first entry of the first column that we find in this way~is
\begin{gather}
\notag
(\,1+u+u\,Y\be-u\,Y\al\,(\ts1+u\,X\ts\al\ts)^{\ts-1}\,u\,X\ts\be\,)^{\ts-1}=
\\
\notag
(\,1+u+u\,Y\be-u^{\ts2}\,Y\al\,X\,(\ts1+u\,\al\,X\ts)^{\ts-1}\ts\be\,)^{\ts-1}=
\\
\notag
(\,1+u+u\,Y\,(1-u\,L\,(\ts1+u\,L\ts)^{\ts-1}\ts)\,\be\,)^{\ts-1}=
\\
\notag
(\,1+u+u\,Y\,(\ts1+u\,L\ts)^{\ts-1}\,\be\,)^{-1}=
\\
\label{entry1}
(\,1+u+u\,J\ts(u)\ts)^{-1}\,.
\end{gather}
The second entry of the first column of the inverse matrix 
that we find is then   
\begin{gather}
\notag
-\,(\ts1+u\,X\ts\al\ts)^{\ts-1}\,u\,X\ts\be\,(\,1+u+u\,J\ts(u)\ts)^{\ts-1}=\ 
\\
\label{entry2}
-\,u\ts\,X\,(\ts1+u\,L\ts)^{\ts-1}\be\,(\,1+u+u\,J\ts(u)\ts)^{\ts-1}\,.
\end{gather}

The product $V\,\ga$ in 
\eqref{iu} is represented by the row with operator entries
$$
\begin{bmatrix}
\,w-1\ts
\,\ts,\,
-\ts Y\,
\end{bmatrix}
\begin{bmatrix}
\ 1&\,0\ 
\\[2pt]
\ \be&\al\ 
\end{bmatrix}
=
\begin{bmatrix}
\,w-1-Y\ts\be
\ ,\,
-\ts Y\al\,
\end{bmatrix}
.
$$
The series $I\ts(u)$ is equal to the product of
this row by the column representing 
$(\,1+u\,Z\ts)^{\ts-1}\,\de\,$. That column
has the entries \eqref{entry1} and \eqref{entry2}.
Hence $I\ts(u)$ equals 
\begin{gather*}
(\ts w-1-Y\ts\be\ts)\ts(\,1+u+u\,J\ts(u)\ts)^{-1}+
Y\al\,u\,X\ts(\ts1+u\,L\ts)^{\ts-1}\be\,(\,1+u+u\,J\ts(u)\ts)^{\ts-1}
\\
=\,(\ts w-1-Y\ts(\ts1-u\,L\,(\ts1+u\,L\ts)^{\ts-1})\,\be\,)\ts
(\,1+u+u\,J\ts(u)\ts)^{\ts-1}
\\
=\,(\ts w-1-Y\ts(\ts1+u\,L\ts)^{\ts-1}\ts\be\,)\ts
(\,1+u+u\,J\ts(u)\ts)^{\ts-1}
\\
=\,(\ts w-1-J\ts(u)\ts)\ts(\,1+u+u\,J\ts(u)\ts)^{\ts-1}\,.
\end{gather*}
Our theorem immediately follows from the last displayed
expression for $I\ts(u)\,$.~\qed
\end{proof}

Note that 
by replacing in the series $-\,u\,J(u)$
the variable $u$ by $-\,u^{\ts -1}$ we get the same
generating series for the limits 
of the quantum Hamiltonians at $N\to\infty$ as was denoted
in \cite[Section 2]{NS4} by $I\ts(u)\,$. But in the present article
the notation $I\ts(u)$ was introduced in \eqref{iu}
and has a meaning different from that in \cite{NS4}.


\section*{\normalsize\bf Acknowledgements}

We are grateful to I.\,V.\,Cherednik for illuminating conversations.
The first named author was supported by the EPSRC grant EP/\ts N\ts023919,
and by the programme 
\emph{Geometry and Representation Theory\/} 
at the Erwin Schr{\"o}dinger Institute.
The second named author was supported by
Leverhulme Senior Research Fellowship.



\end{document}